\documentclass[runningheads]{llncs}
\usepackage{graphicx}
\usepackage{cite}
\usepackage{amsmath,amssymb,amsfonts}

\usepackage{algorithm}
\usepackage{algorithmic}

\usepackage{cuted}
\usepackage{verbatim}
\usepackage{bm}

\usepackage{mathrsfs}
\usepackage{lipsum}
\usepackage{multirow} 
\usepackage[hyphenbreaks]{breakurl}
\newcommand{\tabincell}[2]{\begin{tabular}{@{}#1@{}}#2\end{tabular}}
\begin{document}
\title{Online Maximum $k$-Interval Coverage Problem\thanks{An extended abstract of this paper is to appear in COCOA 2020.}}

%
\author{Songhua Li\inst{1}\and Minming Li\inst{1} \and  Lingjie Duan\inst{2} \and Victor C.S. Lee\inst{1}}
\authorrunning{S. Li et al.}
%
\institute{City University of Hong Kong, Hong Kong SAR, China \\
\email{songhuali3-c@my.cityu.edu.hk\\ \{minming.li,csvlee\}@cityu.edu.hk}
\and
Singapore University of Technology and Design, Singapore
\\
\email{lingjie\_duan@sutd.edu.sg}}
\maketitle              %
\begin{abstract}
We study the online maximum coverage problem on a line, in which, given an online sequence of sub-intervals (which may intersect among each other) of a target large interval and an integer $k$, we aim to select at most $k$ of the sub-intervals such that the total covered length of the target interval is maximized. The decision to accept or reject each sub-interval is made immediately and irrevocably (no preemption) right at the release timestamp of the sub-interval. We comprehensively study different settings of this problem regarding both the length of a released sub-interval and the total number of released sub-intervals. We first present lower bounds on the competitive ratio for the settings concerned in this paper, respectively. For the offline problem where the sequence of all the released sub-intervals is known in advance to the decision-maker, we propose a dynamic-programming-based optimal approach as the benchmark. For the online problem, we first propose a \underline{s}ingle-thresh\underline{o}ld-based deterministic \underline{a}lgorithm SOA by adding a sub-interval if the added length exceeds a certain threshold, achieving competitive ratios close to the lower bounds, respectively. Then, we extend to a \underline{d}ouble-thresh\underline{o}lds-based \underline{a}lgorithm DOA, by using the first threshold for exploration and the second threshold (larger than the first one) for exploitation. With the two thresholds solved by our proposed program, we show that DOA improves SOA in the worst-case performance. Moreover, we prove that a deterministic algorithm that accepts sub-intervals by multi non-increasing thresholds cannot outperform even SOA.
\keywords{Maximum $k$-Coverage Problem\and Budgeted Maximum Coverage Problem\and Interval Coverage \and Online Algorithm}
\end{abstract}
\section{Introduction}
In the classical \textsc{Maximum} $k$-\textsc{Coverage} \textsc{Problem}, we are given a universal set of elements $\textsf{U}=\{U_1,\cdots,U_m\}$ in which each is associated with a weight $w: \textsf{U}\rightarrow \mathbb{R}$, a collection of subsets $\textsf{S}=\{S_1,\cdots,S_n\}$ of \textsf{U} and an integer $k$, and we aim to select $k$ sets from \textsf{S} that maximize the total weight of covered elements in \textsf{U}.
Hochbaum et al. \cite{b35}  showed that this problem is NP-hard and presented a ($1-\frac{1}{e}$)-approximation algorithm that greedily selects a set that maximally increases the current overall coverage. The \textsc{Budgeted} \textsc{Maximum} \textsc{Coverage} (BMC) problem generalizes the classical coverage problem above by further associating each $S_i\in\textsf{S}$ with a cost $c: \textsf{S}\rightarrow \mathbb{R}$ and relaxing the budget $k$ from an integer to a real number, in which the goal is replaced by selecting a sub-collection of the sets in \textsf{S} that maximizes the total weight of the covered elements in \textsf{U} while adhering to the budget $k$. Clearly, the BMC problem is also NP-hard and actually has a $(1-\frac{1}{e})$-approximation algorithm \cite{b37}. In the online version of the above maximum coverage problems, where at each timestamp $i$ a set $S_i\in \textsf{S}$ is released together with its elements and associated values, an algorithm must decide whether to accept or reject each set $S_i$ at its release timestamp $i$ and may also drop previously accepted sets (\textit{preemption}). However, each rejected or dropped set cannot be retrieved at a later timestamp. 

In this paper, we consider the online maximum $k$-coverage problem on a line without preemption. Given an online sequence of sub-intervals of a target interval, we aim to accept $k$ of the sub-intervals irrevocably such that the total covered length of the target interval is maximized.  We refer to this variant as the \textsc{Online} \textsc{Maximum} $k$-\textsc{Interval} \textsc{Coverage} \textsc{Problem} as formally defined in Section \ref{prelinimarysection}. Regarding the length of a sub-interval, we generally consider the Unit-Length (UL), the Flexible-Length (FL), and the Arbitrary-Length (AL) settings, respectively. We consider the Unique-Number (UN) and the Arbitrary-Number (AN) settings, respectively, regarding the total number of released sub-intervals. In particular, our problem under the UN setting is essentially the classical maximum $k$-coverage problem (or say, the BMC with unit-cost sets only and an integer budget $k$) without preemption, by the following \textit{reduction} method: we partition the target interval of our problem into discrete small intervals by the boundary points of all the released sub-intervals, then, the small intervals are equivalent to the elements of a universal set \textsf{U} in which each element has a weight equal to the length of its corresponding small interval, and the released sub-intervals are equivalent to the sets in the collection $\textsf{S}=\{S_1, \cdots, S_n\}$.  The objective remains the same. 

\textbf{Related Works}. We survey relevant researches along two threads. \textit{The first thread} is about the Online Budgeted Maximum Coverage (OBMC) problem, Saha et al. \cite{b33} presented a $4$-competitive deterministic algorithm for the setting where sets have unit costs. Rawitz and Rosén \cite{b32} showed that the competitive ratio of any deterministic online algorithm for the OBMC problem must depend on the maximum ratio $r$ between the cost of a set and the total budget, and also presented a lower bound of $\Omega (\frac{1}{\sqrt{1-r}})$ and a $\frac{4}{1-r}$-competitive deterministic algorithm. Ausiello et al. \cite{b36} studied a special variant of online maximum $k$-coverage problem, the maximum $k$-vertex coverage problem, where each element belongs to exactly two sets and the intersection of any two sets has size at most one. They presented a deterministic 2-competitive algorithm and gave a lower bound of $\frac{3}{2}$.  The \textit{second thread} is about the online $k$-secretary problem \cite{b25,b26,b4}, which was introduced by Kleinberg \cite{b23} and aimed to select $k$ out of $n$ independent values for maximizing the expected sum of individual secretary values. Bateni et al. \cite{b31} studied a more general version called the submodular secretary problem, which aims to maximize the expectation of a submodular function that defines the efficiency of selected candidates based on their overlapping skills. Our problem is similar to theirs as the objective function of our problem is also submodular (see $Len(\cdot)$ of our model in Section 2). However, we focus on the adversarial release order of sub-intervals (secretaries) in the worst-case analysis of deterministic algorithms while \cite{b31} focused on a random release order of secretaries in the average-case analysis of algorithms. Other works related to this paper include the interval scheduling problem, the set cover problem, and the online knapsack problem. Interested readers may refer to \cite{b12,b13,b40,b41,b28,b29}.\\
\begin{table}[htbp]
\begin{center}
\begin{tabular}{c|c|c|c}
\hline\hline
\multicolumn{2}{c|}{\textbf{Settings}} &\textbf{Lower bounds}& \textbf{Upper bounds}\\
\hline
\multirow{2}*{UL} &UN & 
\tabincell{c}{$\sqrt{2}$ for $k=2$\\ 
decrease as $k\geq 3$ increase\\(Theorem \ref{unit-lengthLB1})}&\tabincell{c}{\bm{$<2$}\\
(\textbf{Theorems \ref{FixHitchingratio} \& \ref{doublehitchingcompetitveratio}})}\\
\cline{2-4}
 &AN &
\tabincell{c}{$\sqrt{2}$ for $k=2$\\ 
decrease as $k\geq 3$ increase\\(Corollary \ref{lb_ulan})} &
\tabincell{c}{$\frac{\sqrt{9 k^2-14k+9}-k-1}{2(k-1)}+1$\\ (Corollary \ref{ub_ulan})} \\
 \hline
\multirow{2}*{FL} &UN&
\tabincell{c}{$\frac{2km}{2km+(1-m)  \min\{k,n-k\}}$($<2$) \\(Theorem \ref{varied-lengthLB})} &
\tabincell{c}{\bm{$<1+\frac{k}{k-1}\sqrt{\frac{1+8m}{4}}$}\\(\textbf{Theorem \ref{Varied_FixHitchingratio}})} \\
\cline{2-4}
&AN &\tabincell{c}{$\frac{2m}{m+1}$\\(Corollary \ref{lb_flan})} &\tabincell{c}{ $\frac{\sqrt{(1+8m)  k^2-(6+8m)k+9}-k-1}{2(k-1)}+1$\\(Corollary \ref{ub_flan})} \\
 \hline
AL& UN or AN & $+\infty$ (Theorem \ref{arbitrary-lengthLB}) & - \\
 \hline
US &UN&
\tabincell{c}{$\sqrt{2}$ for $k=2$\\ 
decrease as $k\geq 3$ increase\\(Corollary \ref{unit-sumLB1})}
 &\tabincell{c}{\bm{$<2$}\\(\textbf{Theorem \ref{unit-sumUB1}})}\\
\hline\hline
\end{tabular}
\label{tab1}
\end{center}
\caption{Main results in this Paper}
\end{table}

\textbf{Our contribution}.  Results of this paper are three-fold. \textit{First}, we show that no online deterministic algorithm can achieve a bounded competitive ratio in the AL setting, and present lower bounds on the competitive ratio for the other settings, respectively, in a constructive way. \textit{Second}, we give an $O(kn+n\log n)$-time optimal solution to the offline problem where the sequence of all the released sub-intervals is known in advance to the decision-maker, by applying a dynamic programming-based approach. \textit{Third}, for the online problems, we propose two $O(n)$-time deterministic algorithms, SOA and DOA, with their competitive ratios proved to be close to the lower bounds in the settings, respectively. We also extend our results in UL to a generalized unit-sum (US) setting, where at each timestamp, a batch of a finite number of disjoint sub-intervals is released instead and accordingly one can accept at most $k$ released batches. In addition, we show that any deterministic algorithm, that accepts sub-intervals by non-decreasing thresholds, cannot achieve better performance even than the SOA does. 

Main results of this paper are summarized in Table \ref{tab1}, in which, for ease of understanding, some complicated parameter-dependent results are approximated by formulations in bold. For precise results, please refer to the corresponding theorems or corollaries. 
\section{Preliminaries} \label{prelinimarysection}
\begin{table}[htpb]\label{notatinosinthispaper}
\begin{center}
\begin{tabular}{cp{200pt}}
\hline\hline
 Notations & Descriptions\\
 \hline\hline
 $[0,a]$& The target interval;\\
\hline 
$k$& The maximum number of sub-intervals to accept;\\
\hline
$V_i=[o_i,d_i]$&  The $i$th released sub-interval;\\
\hline
$\mathbb{V}_i=\{V_1,V_2,\cdots,V_i\}$ & The sequence of the first $i$ released sub-intervals;\\
\hline
$\chi(\mathbb{V}_n,k)$ & The optimal solution for the offline problem, given both the set $\mathbb{V}$ of offline sub-intervals and the quota $k$ beforehand;\\
\hline
$\Lambda (V_i,V_j)$ & The length of the intersection between sub-intervals $V_i$ and $V_j$, i.e., $\Lambda (V_i,V_j)=|V_i\cap V_j|$;\\
\hline
$\Phi (\mathbb{V}_i)$ & The subset of $\mathbb{V}_i$ that are 
accepted by our algorithm; \\
\hline
$Len(U)$ &  The cumulative length of the parts of $[0,a]$ that are covered by sub-intervals in a given set $U$, i.e., $Len(U)=|\bigcup_{V_i\in U}V_i|$. Also, we use $Len(V_i)$ to denote the length of a sub-interval $V_i$, i.e., $Len(V_i)=|V_i|$. \\
\hline\hline
\end{tabular}
\end{center}
\caption{Notations in this paper.}\label{tab2}
\end{table}
\textbf{The Model}. Table \ref{notatinosinthispaper} summarizes key notations in this paper.  An online sequence $\mathbb{V}=\{V_1,V_2,\cdots\}$  of sub-intervals of a large target interval $[0,a]$  are released in an adversarial order to the decision-maker, in which  $V_i=[o_i,d_i]\subseteq[0,a]$ for each $V_i\in \mathbb{V}$. Upon the arrival of each $V_i\in \mathbb{V}$, the decision-maker must make a decision whether to accept or reject $V_i$ immediately and irrevocably.\textit{ For example, when recruiting at most $k$ employees across different domains of expertise in the target interval, each released sub-interval represents a candidate's expertise domain. The hiring decision on each sub-interval is irrevocable and must be made on candidate arrival without knowing future sub-intervals.} Due to the quota limitation, the decision-maker can accept no more than $k$ ($\geq2$) sub-intervals \footnote{When $k=1$, our problem degenerates to the classical secretary problem without expertise sub-interval overlap.}. Any two different sub-intervals $V_i,V_j\in \mathbb{V}$  may intersect (i.e., $[o_i,d_j]\cap [o_j,d_j]\neq \varnothing$) considering that the expertise of candidates may overlap in reality. Now, we formally define the settings studied in this paper: with respect to the length $(d_i-o_i)$ of each $V_i\in \mathbb{V}$, we consider three settings.
\begin{itemize}
    \item  \textbf{Unit Length\;(UL)}: $|d_i-o_i|=1$ is normalized with regard to $a$;
    \item \textbf{Flexible Length\;(FL)}: $|d_i-o_i|$ varies in a known range $[1,m]$, in which $m>1$ as $m=1$ degenerates the case to the UL setting;
    \item \textbf{Arbitrary Length\;(AL)}:$|d_i-o_i|$ varies arbitrarily in $[0,a]$;
\end{itemize}
In addition, we also consider a generalized version of the UL setting, which is the \textbf{Unit Sum\;(US)} setting: each $V_i\in \mathbb{V}$ is no longer restricted to contain only one sub-interval, but a batch of a finite number of disjoint sub-intervals of $[0,a]$ whose sum length is equal to 1. This tells that a candidate masters different domains of expertise. We keep the same unit-sum for all the sub-intervals to tell similar strength of all the job candidates. Accordingly, $k$ batches of sub-intervals can be accepted in the US setting. With respect to the number $|\mathbb{V}|$ of total released sub-intervals, we consider the following two settings respectively.
\begin{itemize}
    \item \textbf{Unique Number \;(UN)}: $|\mathbb{V}|$ is known in advance  as a constant $n\in \mathbb{N}^*$. We further restrict $n\geq k+1$ as otherwise (when $n\leq k$) an optimal solution can be easily achieved by just accepting all sub-intervals;
    \item \textbf{Arbitrary Number \;(AN)}: $|\mathbb{V}|$ is not known;
\end{itemize}
When two settings are linked by a "-", we refer to the case that the two settings hold together. For example, we use UL-UN to refer to the setting where all sub-intervals have unit length and the total number of released sub-intervals are known in advance. Whenever we specify a single setting in one dimension, we do not distinguish among settings in the other dimension. For example, when specifying the UN setting only, we actually refer to the context as any setting in \{UL-UN, FL-UN, AL-UN\}. 

Given a sequence $\mathbb{V}=\{V_1,V_2,\cdots\}$ of online sub-intervals of $[0,a]$, the \textbf{objective} is to accept a subset $U\subseteq \mathbb{V}$ of sub-intervals such that $|U|\leq k$ and the cumulative length $Len(U)$ of the parts of $[0,a]$ that are covered by accepted sub-intervals in $U$ is maximized. Denote ALG($\mathbb{V}$) and OPT($\mathbb{V}$) as the covered length by an online algorithm ALG and by an optimal offline solution with complete information of all sub-intervals known beforehand, respectively. We slightly abuse notations by rewriting ALG($\mathbb{V}$) and OPT($\mathbb{V}$) to ALG and OPT, respectively. For $\rho\geq 1$, a deterministic online algorithm ALG is called $\rho$-competitive for the problem if OPT($\mathbb{V}$)$\leq \rho $ALG($\mathbb{V}$) for every instance $\mathbb{V}$. Alternatively, we also say the competitive ratio of ALG is $\rho$ for the problem. Further, when a number $\gamma \geq 1$ ensures that $\gamma  \leq \rho$ holds for all deterministic online algorithms, we say $\gamma $ is a lower bound the on competitive ratio for the problem.
\section{Lower Bounds}
We construct lower bounds on the competitive ratio for the settings studied in this paper, respectively. 
\begin{theorem}\label{arbitrary-lengthLB}
In the AL setting, no online deterministic algorithm can achieve a bounded competitive ratio.
\end{theorem}
\begin{proof}
Let $\varepsilon$ be a small positive number, i.e., $0<\varepsilon<<1$. Suppose the first $k$ sub-intervals released as $\mathbb{V}_k=\{[0,\varepsilon^{k+1-i}]|i=1,2,\cdots,k\}$. We discuss two cases. 
\\
\textbf{Case 1.} Online algorithm (ALG) rejects some sub-interval $V_j=[0,\varepsilon^{k+1-j}]\in \mathbb{V}_k$. Afterwards, the adversary only release sub-intervals as $[0,\varepsilon^{k+1-j+1}]$ instead. This way, the optimal solution (OPT) is able to achieve an overall length at least $\varepsilon^{k+1-j}$ by accepting $V_j$, while ALG can achieve an overall length at most $\varepsilon^{k+1-j+1}$ by sub-intervals in $\mathbb{V}_{j-1}$, we have $\rho\leq \frac{\varepsilon^{k+1-j}}{\varepsilon^{k+1-j+1}}=\frac{1}{\varepsilon}\rightarrow +\infty $ when $\varepsilon\rightarrow 0$;
\\
\textbf{Case 2.} ALG accepts all the $k$ sub-intervals in $\mathbb{V}_k$ and hence runs out of its quota.  Afterward, the adversary only release sub-intervals as $[0,1]$. Then, OPT is able to achieve an overall length 1 by accepting some $[0,1]$, while ALG achieves an overall length exactly equal to $\varepsilon$ by $\mathbb{V}_k$, we have $\rho\leq \frac{1}{\varepsilon}\rightarrow +\infty $ when $\varepsilon\rightarrow 0$.
\end{proof}
\begin{theorem}\label{unit-lengthLB1}
In the UL-UN setting, no online deterministic algorithm can achieve a competitive ratio better than (\ref{ULUN_LB}), in which $\alpha=\left \lfloor 1-\frac{\log (k^{\frac{1}{k}}-1)}{\log (k^{\frac{1}{k}})} \right \rfloor$
\begin{equation}\label{ULUN_LB}
\begin{small}
\begin{cases}
\sqrt{2},&{\rm if\;} k=2\\
\min\{k^{\frac{1}{k}},\frac{k^{\frac{\alpha }{k}}+k-\alpha-1}{k^{\frac{\alpha }{k}}+k-\alpha-2},\frac{k}{k^{\frac{\alpha }{k}}+k-\alpha-1}\},& {\rm if\;} 3\leq k\leq n-\alpha-1\\
\min\{k^{\frac{1}{k}},\frac{k^{\frac{\alpha }{k}}+k-\alpha-1}{k^{\frac{\alpha }{k}}+k-\alpha-2},\frac{n-\alpha+2+k^{\frac{\alpha}{k}}-k^{\frac{n-k}{k}}}{k^{\frac{\alpha }{k}}+k-\alpha-1}\},& {\rm if\;} n-\alpha\leq k\leq n-1\\
\end{cases}
\end{small}
\end{equation}
\end{theorem}
\begin{proof}
Given the number $k$ of quota (i.e., the maximum number of sub-intervals to accept), the number $n$ ($\geq k+1$) of released sub-intervals of the target interval $[0,a]$ with the right endpoint $a$ chosen as a large number, we prove this theorem for $k=2$ and $3\leq k\leq n-1$, respectively. Note that sub-intervals in $\mathbb{V}=\{V_1,\cdots,V_n\}$ arrive in increasing order of their subscripts.
\\
\textbf{Case 1. $k=2$}. \quad We prove the lower bound by the following constructed instance. Considering that $V_1=[0,1]$ and  $V_2=[\sqrt{2}-1,\sqrt{2}]$, we discuss it in the following three cases.
\\
\textbf{Case 1.1.} ALG accepts both $V_1$ and $V_2$. Then, the future ($n-2$) sub-intervals are released as $\{V_j=[1,2]| V_j\in\{V_3,...,V_n\}\}$. This way, $\rho= \frac{2}{\sqrt{2}}=\sqrt{2}$ as OPT can accept $V_1$ and $V_3$.
\\
\textbf{Case 1.2.} ALG accepts $V_2$ and rejects $V_1$. Then, the future ($n-2$) sub-intervals arrive as  $\{V_j=[\sqrt{2}-1,\sqrt{2}]| V_j\in\{V_3,...,V_n\}\}$. Hence, $\rho= \frac{\sqrt{2}}{1}$ as OPT accepts $V_1$ and $V_2$;
\\
\textbf{Case 1.3.} ALG accepts $V_1$ and rejects $V_2$, or ALG rejects both $V_2$ and $V_1$. Then, the future ($n-2$) sub-intervals arrive as $\{V_j=[0,1]| V_j\in\{V_3,...,V_n\}\}$. Hence, $\rho=\frac{\sqrt{2}}{1}$ as OPT accepts $V_1$ and $V_2$.
\\
\textbf{Case 2.} $3\leq k\leq n-1$. We show the lower bound by the following constructive policy: (1) the $n$ sub-intervals in $\{V_1,V_2,V_3,...,V_{n}\}$ are initially supposed to arrive in increasing order of their subscripts (see Figure \ref{LBunite}); (2) if ALG rejects $V_1$, sub-intervals in
$\{V_2,V_3,...,V_n\}$ are replaced by
another ($n-1$) new sub-intervals $\{V'_2,V'_3,...,V'_{n}\}$ with the same range from 1 to 2, i.e., $V'_i=[1,2]$
for each $V'_i\in\{V'_2,V'_3,...,V'_n\}$; (3) if ALG rejects some sub-interval
$V_j\in\{V_2,V_3,...,V_{k}\}$, all the future sub-intervals are replaced by another ($n-j$) new ones 
$\{V'_{j+1},...,V'_{n}\}$ that have the same range as $V_{j-1}$. 
\begin{align*}
    &V_1=[\theta_0,\theta_0+1]\\
&V_2=[\theta_0+\theta_1,\theta_0+\theta_1+1]\\
&\cdots\\
&V_j=[\sum_{i=0}^{j-1}\theta_i,\sum_{i=0}^{j-1}\theta_i+1]\\
&\cdots\\
&V_{n}=[\sum_{i=0}^{n-1}\theta_i,\sum_{i=0}^{n-1}\theta_i+1]
\end{align*}
in which, 
\begin{equation}
\label{theta}
\theta_i=
\begin{cases}
0,&i=0\\
k^{\frac{i}{k}}-k^{\frac{i-1}{k}}, & 1\leq i\leq \alpha\\
1,& \alpha+1\leq i\leq n-1
\end{cases}
\end{equation}
in which $\alpha=\left \lfloor 1-\frac{k  \log (k^{\frac{1}{k}}-1)}{\log k} \right \rfloor$ with $3\leq \alpha\leq k$ for $3\leq k$, and $r_1<r_2<\cdots<r_k$. We let the right-side endpoint $a$ of the target interval $[0,a]$ be larger than $ 1+\sum_{i=0}^{n-1}\theta_i$, which guarantees the above $V_1,\cdots,V_n$ are all sub-intervals of the target interval $[0,a]$. 
By Equation (\ref{theta}), we have
\begin{equation}\label{headtotail}
\sum_{m=0}^{i}\theta_m=1+\sum_{m=0}^{i-1}\theta_m, \forall i\in \{\alpha+1,...,n-1\}
\end{equation} 
i.e., the end point $d_i$ ($=1+\sum_{m=0}^{i-1}\theta_m$) of  each sub-interval $V_i\in\{V_{\alpha+1},...,V_{n-1}\}$ is just the start point $o_{i+1}$ ($=\sum_{m=0}^{i}\theta_m$) of the next sub-interval $V_{i+1}$ to be released. Under the above policy, We show the lower bound by the following three cases. 
\begin{figure}
    \centering
    \includegraphics[width=8cm]{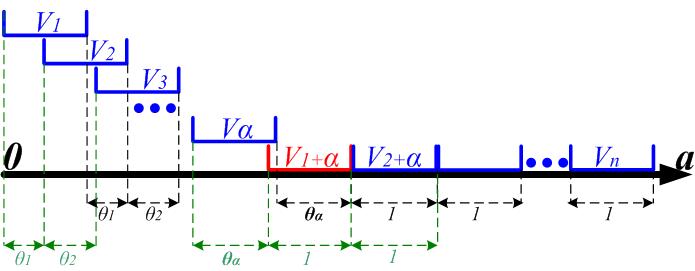}
    \caption{Configuration of the sub-intervals in the lower bound of the unit-length case.}
    \label{LBunite}
\end{figure}
\\
\textbf{Case 2.1.} ALG rejects $V_1$. ALG can achieve an overall covered length at most 1 from future sub-intervals in $\{V'_2,...,V'_n\}$, implying $\rho=2$ as OPT can get an overall length of 2 by accepting both $V_1$ and another $k-1$ sub-intervals in $\{V'_2,...,V'_n\}$. 
\\
\textbf{Case 2.2.} ALG accepts $\{V_1,V_2,...,V_{j-1}\}$ and rejects $V_j$ ($\in\{V_2,...,V_{k}\}$). As all the future sub-intervals have the same range as $V_{j-1}$, ALG can get an overall length of $1+\sum_{i=0}^{j-2}\theta_i$ by the first ($j-1$) sub-intervals $\{V_1,...,V_{j-1}\}$. In contrast, OPT can get an overall length of $1+\sum_{i=0}^{j-1}\theta_i$ by accepting all the first $j$ sub-intervals in $\{V_1,...,V_{j}\}$, implying the ratio $\rho=\frac{1+\theta_1+\cdots+\theta_{j-2}+\theta_{j-1}}{1+\theta_1+\cdots+\theta_{j-2}}$.  We further discuss the ratio in two sub-cases.
\\
\textbf{Case 2.2.1.} $j\leq \alpha$ ($=\left \lfloor 1-\frac{k  \log (k^{\frac{1}{k}}-1)}{\log k} \right \rfloor$). By Equation (\ref{theta}),
\begin{center}
$\rho=\frac{1+(k^{\frac{1}{k}}-1)+\cdots+(k^{\frac{j-1}{k}}-k^{\frac{j-2}{k}})+(k^{\frac{j}{k}}-k^{\frac{j-1}{k}})}{1+(k^{\frac{1}{k}}-1)+\cdots+(k^{\frac{j-1}{k}}-k^{\frac{j-2}{k}})}=\frac{k^{\frac{j}{k}}}{k^{\frac{j-1}{k}}}=k^{\frac{1}{k}}<1.5<2$.\end{center}
\textbf{Case 2.2.2.} $\alpha+1\leq j\leq k$. By Equation (\ref{theta}), we have
\begin{center}
$\rho=\frac{1+(k^{\frac{1}{k}}-1)+\cdots+(k^{\frac{\alpha}{k}}-k^{\frac{\alpha-1}{k}})+j-\alpha-1}{1+(k^{\frac{1}{k}}-1)+\cdots+(k^{\frac{\alpha}{k}}-k^{\frac{\alpha-1}{k}})+j-\alpha-2}=\frac{k^{\frac{\alpha}{k}}+j-\alpha-1}{k^{\frac{\alpha}{k}}+j-\alpha-2}\geq \frac{k^{\frac{\alpha}{k}}+k-\alpha-1}{k^{\frac{\alpha}{k}}+k-\alpha-2}$\end{center}
in which the inequality holds by the basic condition of this case.
\\
\textbf{Case 2.3.} ALG accepts the first $k$ sub-intervals $\{V_1,V_2,V_3,...,V_{k}\}$. This implies ALG gets an overall length of ($k^{\frac{\alpha}{k}}+k-\alpha-1$). Later, the future sub-intervals are released as  $\{V_{k+1},...,V_{n}\}$. We then discuss two cases.
\\
\textbf{Case 2.3.1.} $\alpha+k+1\leq n$. By  Equation (\ref{headtotail}) and $n-k+1\geq \alpha+2$, we know OPT can get an overall length of $k$ by the last $k$ sub-intervals released, i.e.,  $\{V_{n-k+1},...,V_n\}$. This implies the ratio $\rho=\frac{k}{k^{\frac{\alpha}{k}}+k-\alpha-1}$.
\\
\textbf{Case 2.3.2.} $1+k\leq n \leq \alpha+k$. By Equation (\ref{theta}), we have $\theta_0<\cdots<\theta_{\alpha+1}=\cdots=\theta_n=1$. As OPT performs no worse than accepting the last $k$ sub-intervals $\{V_{n-k+1},...,V_{n}\}$, we know OPT can get an overall length no less than
\begin{equation*}
\begin{split}
&Len(\bigcup_{i=n-k+1}^{n}V_i)\\
&=Len(\bigcup_{i=\alpha+2}^{n}V_i)+Len(\bigcup_{i=n-k+1}^{\alpha+1}V_i)\\
&=Len(\bigcup_{i=\alpha+2}^{n}V_i)+Len(\bigcup_{i=1}^{\alpha+1}V_i)-\sum^{n-k}_{i=1}\theta_i\\
&=(n-\alpha+1)+k^{\frac{\alpha}{k}}-\sum_{i=1}^{n-k}\theta_i\\
&=n-\alpha+2+k^{\frac{\alpha}{k}}-k^{\frac{n-k}{k}}
\end{split}
\end{equation*}
in which the first equation holds by Equation (\ref{headtotail}) and the second equation holds by Equation (\ref{theta}) (see Figure \ref{LBunite}). This further implies the ratio \begin{center}
$\rho\geq
\frac{n-\alpha+2+k^{\frac{\alpha}{k}}-k^{\frac{n-k}{k}}}{k^{\frac{\alpha}{k}}+k-\alpha-1}
=\frac{k^{\frac{\alpha}{k}}+k-\alpha-1+(n-k+3-k^{\frac{n-k}{k}})}{k^{\frac{\alpha}{k}}+k-\alpha-1}$.\end{center}

Therefore, no online algorithm can beat a competitive ratio
\begin{center}
$\min\{k^{\frac{1}{k}},\frac{k^{\frac{\alpha }{k}}+k-\alpha-1}{k^{\frac{\alpha }{k}}+k-\alpha-2},\frac{\beta}{k^{\frac{\alpha }{k}}+k-\alpha-1}\}$\end{center}
in which $\beta=\min\{k,n-\alpha+2+k^{\frac{\alpha}{k}}-k^{\frac{n-k}{k}}\}$.
\end{proof}
\begin{corollary}\label{lb_ulan}
For UL-AN, no online deterministic algorithm can achieve a competitive ratio better than (\ref{ULAN_LB}), in which  $\alpha=\left \lfloor 1-\frac{\log (k^{\frac{1}{k}}-1)}{\log (k^{\frac{1}{k}})} \right \rfloor$.
\begin{equation}\label{ULAN_LB}
\begin{small}
\begin{cases}
\sqrt{2},&{\rm if\;} k=2\\
\min\{k^{\frac{1}{k}},\frac{k^{\frac{\alpha }{k}}+k-\alpha-1}{k^{\frac{\alpha }{k}}+k-\alpha-2},\frac{k}{k^{\frac{\alpha }{k}}+k-\alpha-1}\},& {\rm if\;} 3\leq k\\
\end{cases}
\end{small}
\end{equation}
\end{corollary}
\begin{proof}
Note that in UL-AN, online algorithm does not know the number ($|\mathbb{V}|$) of all the sub-intervals to be released in advance and can only learn a sub-interval upon its release timestamp. This implies that the adversary can control $|\mathbb{V}|$ to release an arbitrary number of sub-intervals. By using a similar release policy of sub-intervals as in the proof of Theorem 2 with the $n=|\mathbb{V}|$ discarded, the sub-case 2.3.2 of Theorem 2 is further removed since the total covered length achieved by OPT is bounded just by quota $k$ in this unit-length setting, while the other cases of Theorem 2 does not change significantly. Hence, we have the lower bound of UL-AN as
\begin{equation*}
\begin{cases}
\sqrt{2},&{\rm if\;} k=2\\
\min\{k^{\frac{1}{k}},\frac{k^{\frac{\alpha }{k}}+k-\alpha-1}{k^{\frac{\alpha }{k}}+k-\alpha-2},\frac{k}{k^{\frac{\alpha }{k}}+k-\alpha-1}\},& {\rm if\;} 3\leq k\\
\end{cases}\end{equation*}
in which $\alpha=\left \lfloor 1-\frac{k  \log (k^{\frac{1}{k}}-1)}{\log k} \right \rfloor$.
\end{proof}
\begin{theorem} \label{varied-lengthLB}
For FL-UN, no online deterministic algorithm can achieve a competitive ratio better than $\frac{2km}{2km+(1-m)  \min\{k,n-k\}}$ which is strictly smaller than 2.
\end{theorem}
\begin{proof}
In FL-UN, the length of each sub-interval $V_i=[o_i,d_i]$ belongs to a known range $[1,m]$, i.e., $|d_i-o_i|\in [1,m]$. We show the lower bound by the following constructive policy: the first $\tau =\min\{k,n-k\}$ sub-intervals are released as $\{V_i=[i-1,i]| i\in\{1,2,...,\tau \}\}$. Suppose ALG accepts $x\in \mathbb{N}$ out of the $\tau$ sub-intervals.
\\
\textbf{Case 1.} $0<x\leq \left \lfloor \frac{\tau}{2} \right \rfloor$. Then, all the future sub-intervals have the same range $[0,1]$. This way, ALG can get an overall length at most $\left \lfloor \frac{\tau}{2} \right \rfloor$, implying $\rho\geq \frac{\tau}{\left \lfloor \frac{\tau }{2} \right \rfloor}\geq 2$;
\\
\textbf{Case 2.} $\left \lceil \frac{\tau }{2} \right \rceil\leq x\leq \tau$. Then, the remaining ($n-\tau$) sub-intervals arrive as 
\begin{equation*}
\{V_j=[\tau+(j-\tau -1)  m,\tau +(j-\tau )  m]|j\in\{\tau+1,\tau +2,...,n\}\}.\end{equation*}
This way, ALG can get an overall length at most $x+(k-x)  m$, which is by accepting $x$ out of the first $\tau$ sub-intervals and ($k-x$) out of the last ($n-\tau$) sub-intervals.  Since $n-\tau\geq k$, OPT is able to get an overall length  $k m$ by accepting $k$ out of $n-\tau$ sub-intervals in $\{V_{\tau +1},...,V_{\tau +k}\}$. Hence, 
\begin{small}
\begin{equation*}
\rho\geq\frac{k  m}{km+x(1-m)}\geq \frac{2km}{2km+(1-m)  \min\{k,n-k\}}.\end{equation*}
\end{small}
Therefore, no online algorithm can beat a competitive ratio of 
\begin{center}
$\min\{2,\frac{2km}{2km+(1-m)  \min\{k,n-k\}}\}=\frac{2km}{2km+(1-m)  \min\{k,n-k\}}$\end{center}
\end{proof}
\begin{corollary}\label{lb_flan}
For FL-AN, no online deterministic algorithm can achieve a competitive ratio better than $\frac{2m}{m+1}$ which is strictly smaller than 2.
\end{corollary}
\begin{proof}
In FL-AN, the total number  $|\mathbb{V}|$ of sub-intervals is not limited to $n$ any more. By setting  $\tau=k$ and removing  $n=|\mathbb{V}|$ in the release policy of sub-intervals in Theorem \ref{varied-lengthLB}, 
the overall length of OPT is bounded by $km$ (in which the $k$ is the quota constraint and $m$ denotes the largest length of a released sub-interval in $\mathbb{V}$) only. Further, We get the lower bound of FL-AN as $\frac{km}{km+\frac{k}{2}(1-m)}=\frac{2m}{m+1}$.
\end{proof}
\begin{corollary}\label{unit-sumLB1}
For US-UN,  no online deterministic algorithm can achieve a competitive ratio better than (\ref{ULUN_LB}), where  $\alpha=\left \lfloor 1-\frac{\log (k^{\frac{1}{k}}-1)}{\log (k^{\frac{1}{k}})} \right \rfloor$.
\end{corollary}
\begin{proof}
By partitioning each sub-interval $V_i\in \mathbb{V}$ of the unit-length case arbitrarily into a finite number of disjoint sub sub-intervals, we can get an instance of the unit-sum case. Hence, the lower bound, which is showed in Theorem \ref{unit-lengthLB1}, applies to the unit-sum case as well. 
\end{proof}
\section{Upper bounds}
We present two online deterministic algorithms in subsections \ref{soasection} and \ref{doasection} respectively. Before that, we give an $O(kn+n\log n)$ time dynamic programming approach as a benchmark, which optimally solves the offline problem where the sequence of all the released sub-intervals are given beforehand.
\subsection{Dynamic Programming Based Optimal Offline Solution}\label{offlinesolution}
Since both the UL and the FL settings are special cases of the AL setting, we present our offline solution in the AL setting\footnote{We do not distinguish our offline solution in the other dimension since our solution performs optimally in either the UN or the AN.}. Suppose, without loss of generality, that the total number of released sub-intervals in the offline problem equals $n$.
\\
\textit{First.} We sort sub-intervals in $\mathbb{V}_n=\{V_1,V_2,\cdots,V_n\}$ in non-decreasing order of their end locations (i.e., the $d_i$ of each $V_i$), which runs in\textit{ $O(n\log n)$ time}. We abuse notations, in this offline solution only, to denote $(V_1,V_2,\cdots,V_n)$ as the sequence of sorted sub-intervals, i.e., $d_1\leq d_2\leq \cdots\leq d_n$, and further $\mathbb{V}_i=\{V_1,\cdots,V_i\}$ as the first $i$ sub-intervals in the sequence. Suppose the decision-maker accepts sub-intervals in $\mathbb{V}_n$ in decreasing order of their subscripts as well.
\begin{definition}
 $ V_{\psi(i)}=\mathop{\arg\max}\limits_{\{V_j\in \mathbb{V}_{i-1}| o_j<o_i\leq d_j\}}\{o_i-o_j\}$ indicates the sub-interval in $\mathbb{V}_{i-1}$ that intersects with $V_i$ and has the left-most start location.
\end{definition}
\begin{definition}
 $V_{\phi(i)}=\mathop{\arg\min}\limits_{\{V_j\in \mathbb{V}_{i-1}| d_j<o_i\}}\{o_i-d_j\}$ indicates the sub-interval in $\mathbb{V}_{i-1}$ that is disjoint from but is closest to $V_i$.
\end{definition}
\begin{proposition}\label{dpforarbitrarylenth}
Once an offline OPT accepts $V_i$, OPT accepts either $ V_{\psi(i)}$ or a sub-interval in $\{V_1,V_2,\cdots,V_{\phi (i)}\}$. 
\end{proposition}
\begin{proof}
When $V_i$ is accepted by offline OPT, the next sub-interval to accept (denoted as $V_{\leftarrow,i}$) lies in set $\mathbb{V}_{i-1}=\{V_{1},\cdots,V_{i-1}\}$ as OPT is supposed to accept sub-intervals in decreasing order of their subscripts. Obviously, $V_{\leftarrow,i}$ either intersects or does not intersect with $V_{i}$. When $V_{\leftarrow,i}$ intersects with $V_i$, we know $V_{\leftarrow,i}$ should be the sub-interval in $\mathbb{V}_{i-1}$ that contributes the most additional length to $V_i$, which is $V_{\psi(i)}$ in Proposition 1, as accepting any sub-interval in $\{V_{\psi(i)},\cdots,V_{i-1}\}$ cannot increase the overall length of OPT; As the sub-interval $V_{\phi}$ in Proposition 1 indicates the sub-interval in $\mathbb{V}_{i-1}$ that is the nearest one to $V_i$ and not intersect with $V_i$, we know $V_{\leftarrow,i}$ lies in $\{V_1,V_2,\cdots,V_{\phi (i)}\}$ when $V_{\leftarrow,i}$ does not intersect with $V_i$.
\end{proof}
\textit{Second}. Since OPT, denoted as $\chi (\mathbb{V}_n,k)$, accepts sub-intervals in $\mathbb{V}_n$ in decreasing order of their subscripts as well, we write the \textit{Bellman Equation} in our dynamic programming  as (\ref{arbitraryDP_basestep1}) and (\ref{arbitraryDP_basestep2}) by setting $i=n$ and $j=k$ initially. Specifically, we discuss the following cases when handling an arbitrary $V_i\in \mathbb{V}_i$.
\begin{enumerate}
    \item OPT rejects $V_i$. Then, we have $\chi (\mathbb{V}_i,j)=Len(\mathbb{V}_i)$ if OPT has enough quota, i.e,  $i\leq j$, to accept all sub-intervals in $\mathbb{V}_i$; or $\chi (\mathbb{V}_i,j)=\chi(\mathbb{V}_{i-1},j)$ otherwise;
    \item OPT accepts $V_i$ and hence runs out of quota ($j=0$). Then, $\chi (\mathbb{V}_i,j)=0$;
    \item OPT accepts $V_i$ and remains quota ($j\geq 1$). By Proposition \ref{dpforarbitrarylenth},
    \begin{enumerate}
        \item OPT further accepts someone in $\{V_1,V_2,\cdots,V_{\phi (i)}\}$. Since $V_i$ is disjoint from the next accepted sub-interval, $\chi (\mathbb{V}_i,j)=Len(V_i)+\chi (\mathbb{V}_{\phi (i)},j-1)$;
        \item  OPT  further accepts $V_{\psi(i)}$. To calculate $\chi (\mathbb{V}_i,j)$, we introduce an intermediate function $\kappa(\mathbb{V}_i,j)$ given in Equation (\ref{arbitraryDP_basestep2})\footnote{The major difference between $\kappa(\mathbb{V}_i,j)$ and $\chi(\mathbb{V}_i,j)$ is that $\kappa(\mathbb{V}_i,j)$ always accepts the last sub-interval $V_i$ in $\mathbb{V}_i$ while $\chi(\mathbb{V}_i,j)$ does not necessarily. }, which always accepts the last sub-interval $V_i$ in $\mathbb{V}_i$ and totally accepts $j$ out of $i$ sub-intervals in $\mathbb{V}_i$ such that the overall covered length of the interval $[0,a]$ is maximized. Then, we count the length contributed by $V_i$ as the part without intersection with $V_{\psi(i)}$, which is $Len(V_i)-\Lambda (V_i,V_{\psi(i)})$, and transit the remaining part of OPT's overall length to $\kappa (\mathbb{V}_{\psi(i)},j-1)$. This way,
    $\chi (\mathbb{V}_i,j)=Len(V_i)-\Lambda (V_i,V_{\psi(i)})+\kappa (\mathbb{V}_{\psi(i)},j-1)$. 
    \end{enumerate}
\end{enumerate}
\begin{equation}\label{arbitraryDP_basestep1}
\begin{split}
   \chi (\mathbb{V}_i,j)=\left\{\begin{matrix}
Len(\mathbb{V}_i),& i\leq j \\ 
\begin{split}
\max\{&\chi(\mathbb{V}_{i-1},j), Len(V_i)+\chi (\mathbb{V}_{\phi (i)},j-1),\\&Len(V_i)-\Lambda (V_i,V_{\psi(i)})+\kappa (\mathbb{V}_{\psi(i)},j-1)\}   
\end{split},& 1\leq j< i
\\
0, & j=0 \\
\end{matrix}\right.
\end{split}
\end{equation}
\begin{equation}\label{arbitraryDP_basestep2}
\begin{split}
    \kappa(\mathbb{V}_i,j)=\left\{\begin{matrix}
           \max\{\begin{split}
           &Len(V_i)+\chi (\mathbb{V}_{\phi(i)},j-1),\\ &Len(V_i)-\Lambda(V_i,V_{\psi(i)})+\kappa(\mathbb{V}_{\psi(i)},j-1)\end{split}\}  ,&j>1\\
            Len(V_i),& j=1 \\ 
            0, & j=0 \end{matrix}\right.
\end{split}
\end{equation}
Note that our dynamic programming solution totally generates $O(kn)$ intermediate states in which each state runs in $O(1)$ time. Together with the preliminary sorting step, our offline solution totally runs in $O(kn+n\log n)$ time.
\subsection{Single-threshold Online Algorithm}\label{soasection}
We first propose an online algorithm, named the
\underline{S}ingle-threshold based \underline{O}nline \underline{A}lgorithm (SOA), for the UN setting. Then, we extend SOA to SOA$_{\rm AN}$ to tackle the AN setting. Note that SOA and SOA$_{\rm AN}$ can achieve competitive ratios strictly smaller than 2 for the UN and the AN settings, respectively. 
\begin{algorithm}[tb] 
\caption{Single-threshold Online Algorithm (SOA)}
\label{FixHitching}
\textbf{Input}:  A sequence $\mathbb{V}=\{V_1,V_2,...,V_n\}$ of $n$ sub-intervals of the target interval $[0,a]$, in which $V_i=[o_i,d_i]$ for each $V_i\in \mathbb{V}$, the quota $k$ ($2\leq k\leq n-1$);\\
\textbf{Output}: A set of accepted sub-intervals, i.e., $\Phi (\mathbb{V}_{n})$;\\
\begin{algorithmic}[1] 
\STATE $\Phi (\mathbb{V}_1)=\{V_1\}$;\qquad\qquad\qquad \qquad\qquad\qquad \qquad\qquad\qquad\qquad \;\quad\COMMENT{always accept $V_1$}\\
\FOR{$i=2;i++;i\leq n$} 
\IF{$|\Phi (\mathbb{V}_{i-1})|=k$} 
\STATE $\Phi (\mathbb{V}_{n})=\Phi  (\mathbb{V}_{i-1})$;
\STATE\textbf{ break};  \quad\quad\qquad\qquad\quad \COMMENT{Complete accepting as SOA runs out of the quota}
\ELSIF{$k-|\Phi (\mathbb{V}_{i-1})|\geq n-i+1$}
\STATE $\Phi (\mathbb{V}_{i})=\Phi (\mathbb{V}_{i-1})\cup V_i$;   \qquad\qquad\qquad \COMMENT{accept $V_i$ by the quota-enough condition}
\ELSE
\IF {$Len(\Phi (\mathbb{V}_{i-1})\cup V_{i})-Len(\Phi (\mathbb{V}_{i-1}))\geq \theta$ with $\theta$ given
 in (\ref{fixedthreshold})} 
\STATE $\Phi (\mathbb{V}_{i})=\Phi (\mathbb{V}_{i-1})\cup V_i$; \qquad\quad\;  \COMMENT{accept $V_i$ by threshold-meeting condition}
\ELSE
\STATE $\Phi (\mathbb{V}_{i})=\Phi (\mathbb{V}_{i-1})$; \qquad\qquad\qquad\qquad\qquad\qquad\;\qquad\qquad\qquad\quad\;\;\COMMENT{reject $V_i$}
\ENDIF
\ENDIF
\ENDFOR
\end{algorithmic}
\end{algorithm}

In the UN setting, SOA always accepts the first released sub-interval $V_1$. On the arrival of each future sub-interval $V_{i}\in \{V_2,...,V_n\}$, SOA accepts $V_{i}$ if and only if it meets one of the following two conditions:  (\romannumeral1) \textbf{Quota-enough condition}, after accepting $V_{i}$, SOA has enough quota to accept all the future sub-intervals, i.e., $k-|\Phi (\mathbb{V}_{i-1})|\geq n-i+1$; (\romannumeral2) \textbf{Threshold-accepting condition}, SOA still has quota (i.e., $|\Phi (\mathbb{V}_{i-1})|\leq k-1$) and $V_i$ contributes an additional length of at least 
\begin{small}
\begin{equation}\label{fixedthreshold}
    \theta=\min\{\frac{\sqrt{1+2(k-1)(n-k)}-1}{2k-2},\frac{\sqrt{9 k^2-14k+9}-k-1}{4(k-1)}\}
\end{equation}\end{small}
\\to the covered length of $[0,a]$ by previously accepted sub-intervals, i.e.,
\begin{equation}\label{fixedcondition}
  Len(\Phi (\mathbb{V}_{i-1})\cup V_{i})-Len(\Phi (\mathbb{V}_{i-1}))\geq \theta
\end{equation}
We summarize SOA in Algorithm \ref{FixHitching} and we note that
\begin{itemize}
    \item Once some sub-interval is accepted by the quota-enough condition, all later-released sub-intervals are accepted by SOA; \\
    \item SOA always uses up its quota to accept $k$ sub-intervals and only breaks (in the Step 5 of Algorithm \ref{FixHitching}) when it accepts $k$ sub-intervals according to the threshold  $\theta$ from the first ($n-1$) released sub-intervals. 
\end{itemize}
\begin{proposition}\label{fixthresholddetermine} In SOA, we have threshold $\theta=\frac{\sqrt{1+2(k-1)(n-k)}-1}{2k-2}$ if $\left \lceil \frac{667n}{1000} \right \rceil \leq k \leq n-1$, or $\theta=\frac{\sqrt{9 k^2-14k+9}-k-1}{4(k-1)}$ if $2\leq k\leq \left \lceil \frac{667n}{1000} \right \rceil-1$.
\end{proposition}
\begin{proof}
First, we notice the following two inequalities by calculations,
\begin{center}
$\frac{\sqrt{1+2(\left \lceil \frac{667n}{1000} \right \rceil-1)(n-\left \lceil \frac{667n}{1000} \right \rceil)}-1}{2\left \lceil \frac{667n}{1000} \right \rceil-2}\leq \frac{\sqrt{9  \left \lceil \frac{667n}{1000} \right \rceil^2-14\left \lceil \frac{667n}{1000} \right \rceil+9}-\left \lceil \frac{667n}{1000} \right \rceil-1}{4(\left \lceil \frac{667n}{1000} \right \rceil-1)}$
\end{center}
and
\begin{center}
$\frac{\sqrt{1+2(\left \lceil \frac{667n}{1000} \right \rceil-2)(n-\left \lceil \frac{667n}{1000} \right \rceil+1)}-1}{2\left \lceil \frac{667n}{1000} -1\right \rceil-2}\geq \frac{\sqrt{9  (\left \lceil \frac{667n}{1000} \right \rceil-1)^2-14(\left \lceil \frac{667n}{1000} \right \rceil-1)+9}-\left \lceil \frac{667n}{1000} \right \rceil-2}{4(\left \lceil \frac{667n}{1000} -1\right \rceil-1)}.$\end{center}

Further, we know
$\frac{\sqrt{1+2(k-1)(n-k)}-1}{2k-2}$ decreases as $k$ increases within $\{2,3,...,n-1\}$, while $\frac{\sqrt{9  k^2-14k+9}-k-1}{4(k-1)}$ increases as $k$ increases, given the number $n$ of online sub-intervals. Then, we get the proposition 2.  
\end{proof}
\begin{theorem} \label{FixHitchingratio}
For UL-UN, SOA runs in $O(n)$ time and achieves a competitive ratio no larger than
$\min\{\frac{\sqrt{1+2(k-1)(n-k)}-1}{k-1}+1,\frac{\sqrt{9 k^2-14k+9}-k-1}{2(k-1)}+1\}$.
\end{theorem}
\begin{proof}
SOA runs in $O(n)$ time as it runs in no more than $n$ iterations in which each iteration runs in $O(1)$ time. To show the upper bound of SOA, we discuss in the following two cases.
\\ \textbf{Case 1}.  SOA accepts $V_n$ (the last released sub-interval).

This shows that SOA triggers the quota-enough condition when accepting some $V_i\in\{V_2,V_3,...,V_n\}$, i.e., $k-|\Phi (\mathbb{V}_{i-1})|\geq n-i+1$. Then, the algorithm accepts all the sub-intervals $\{V_i,V_{i+1},...,V_n\}$ that are released later than $V_i$. Further,
$Len(\Phi (\mathbb{V}_{n}))=Len(\Phi (\mathbb{V}_{i-1})\cup \{V_i,V_{i+1}...,V_{n}\})$. In the worst case, none of the accepted sub-intervals in $\{V_i,V_{i+1},...,V_{n}\}$ contributes additional length to the algorithm since these accepted sub-intervals are also available to OPT. 
Suppose, without loss of generality, that $\bigcup\limits_{V\in \Phi (\mathbb{V}_{i-1})}V$ consists of a number \textbf{$x$} of disjoint intervals, which are denoted by $\mathbb{A}_1$, $\mathbb{A}_2$,..., $\mathbb{A}_x$ respectively. Clearly, $1\leq x\leq k$. 
Namely, $Len(\mathbb{A}_1)$,..., $Len(\mathbb{A}_x)$ respectively denote the length of disjoint intervals of $\Phi (\mathbb{V}_{i-1})$. Hence, $ Len(\Phi (\mathbb{V}_{i-1}))=\sum_{i=1}^xLen(\mathbb{A}_i)$.

Note that each rejected sub-interval can contribute an additional length no more than $\theta$ to $Len(\Phi (\mathbb{V}_{i-1}))$, as otherwise it would have been accepted. On one hand, $ Len(OPT)\leq Len(\Phi (\mathbb{V}_{n}))+ \theta (n-k)$ holds naturally since SOA totally rejects ($n-k$) sub-intervals in $\mathbb{V}_{i-1}$. On the other hand, $Len(OPT)\leq Len(\Phi (\mathbb{V}_{n}))+ 2 \theta  x$ because there are totally $x$ disjoint intervals formed by the sub-intervals accepted by SOA, which implies there are at most $2x$ chances that sub-interval could be missed/rejected by SOA, and each rejected sub-interval can contribute less than $\theta$ to SOA (by Step 9 of SOA). In summary, the  overall length achieved by the OPT is bounded by the following Inequality (\ref{OPTUBCASE1fixed01}).
\begin{equation}\label{OPTUBCASE1fixed01}
    Len(OPT)\leq Len(\Phi (\mathbb{V}_{n}))+ \min\{2 \theta  x,\theta (n-k)\}
\end{equation}
Hence, we get the ratio 
\begin{small}
\begin{flalign*}
\rho&=\frac{Len(OPT)}{Len(\Phi (\mathbb{V}_{n}))}\leq 1+\frac{\min\{2 \theta x,\theta  (n-k)\}}{Len(\Phi (\mathbb{V}_{n}))}\leq 1+\frac{2 \theta  x}{\sum_{i=1}^x Len(\mathbb{A}_i)}&\\
&\leq 1+\min\{\frac{\sqrt{1+2(k-1)(n-k)}-1}{k-1},\frac{\sqrt{9 k^2-14k+9}-k-1}{2(k-1)}\} &
\end{flalign*}\end{small}
\\
in which the first inequality holds by (\ref{OPTUBCASE1fixed01}), the second inequality holds by $Len(\Phi(\mathbb{V}_{n}))\geq Len(\Phi(\mathbb{V}_{i-1}))=\sum_{i=1}^xLen(\mathbb{A}_i)$, and the last inequality holds by $\sum_{i=1}^x Len(\mathbb{A}_i)\geq x$ and (\ref{fixedthreshold}).
\\
\textbf{Case 2}.  SOA does not accept $V_n$. 

This implies the quota-enough condition is not triggered during the execution and SOA accepts $k$ sub-intervals by the threshold-accepting condition. This implies the following Inequality (\ref{SOAoveralllength}) because each accepted sub-interval, except $V_1$ (which contributes 1 to SOA), contributes at least $\theta$ to SOA, see Step 9 of SOA.
\begin{equation}\label{SOAoveralllength}
 Len(\Phi (\mathbb{V}_{n}))\geq 1+(k-1)\theta
\end{equation}
Suppose that $V_i\in\{V_k,...,V_{n-1}\}$ is the last accepted sub-interval by SOA, i.e., $|\Phi(\mathbb{V}_i)|=k$ and
$Len(\Phi(\mathbb{V}_n))=Len(\Phi(\mathbb{V}_{i}))$. In other words, SOA misses all sub-intervals in $\{V_{i+1},...,V_n\}$ which can be accepted by OPT. Since the algorithm can miss at most $n-k$ sub-intervals, OPT can get an accumulating length at most $n-k$ more than that accepted by SOA, i.e., $Len(OPT)\leq Len(\Phi(\mathbb{V}_n))+n-k$. Also, OPT cannot get a length over its quota $k$ in this unit-length case. In summary, the overall length accepted by the OPT is bounded by (\ref{OPTUBINCASE2}). 
\begin{equation}\label{OPTUBINCASE2}
    Len(OPT)\leq \min\{k,Len(\Phi(\mathbb{V}_n))+n-k\}
\end{equation}
We further discuss two sub-cases.
\\
\textbf{Case 2.1}.  $\left \lceil \frac{667n}{1000} \right \rceil \leq k \leq n-1$.
We have $\theta =\frac{\sqrt{1+2(k-1)(n-k)}-1}{2k-2}$ by Proposition \ref{fixthresholddetermine}. Note that
$\frac{\partial (\frac{\theta}{\frac{2k-n-1}{k-1}})}{\partial k}=\frac{8k^2-(8n+10)k+9n-3}{4(2k-n-1)^2\sqrt{1+2(k-1)(n-k)}}<0$ for each $k\in [\left \lceil \frac{667n}{1000}\right \rceil,n]\subseteq (\frac{8n+10-\sqrt{(8n-8)^2+132}}{16},\frac{8n+10+\sqrt{(8n-8)^2+132}}{16})$. Further, we have
\begin{equation}\label{conditionincase21iFixratio}
    \begin{small}
    \frac{\theta}{\frac{2k-n-1}{k-1}}\leq \frac{\theta}{\frac{2k-n-1}{k-1}}|_{k=\left \lceil \frac{667n}{1000} \right \rceil}<1
    \end{small}
\end{equation}
Hence, 
\begin{small}
\begin{flalign*}
    &\rho=\frac{Len(OPT)}{Len(\Phi (\mathbb{V}_{n}))} &\\
    &\leq \min\{\frac{k}{Len(\Phi (\mathbb{V}_{n}))},1+\frac{n-k}{Len(\Phi (\mathbb{V}_{n}))}\} &{\rm by}\; (\ref{OPTUBINCASE2})\\
    &\leq \min\{\frac{k}{1+(k-1)\theta},1+\frac{n-k}{1+(k-1)\theta}\} &{\rm by}\;(\ref{SOAoveralllength})\\
    &= \frac{\sqrt{1+2(k-1)(n-k)}-1}{k-1}+1  &{\rm by}\;(\ref{conditionincase21iFixratio})\;{\rm and\;}\theta =\frac{\sqrt{1+2(k-1)(n-k)}-1}{2k-2} 
\end{flalign*}\end{small}
\\
\textbf{Case 2.2.}  $2\leq k\leq \left \lceil \frac{667n}{1000} \right \rceil-1$. We have $\theta=\frac{\sqrt{9  k^2-14k+9}-k-1}{4(k-1)}\leq \frac{\sqrt{1+2(k-1)(n-k)}-1}{2k-2}$ by Proposition \ref{fixthresholddetermine}. Hence,
\begin{small}
\begin{flalign*}
    &\rho=\frac{Len(OPT)}{Len(\Phi (\mathbb{V}_{n}))} &\\
    &\leq \min\{\frac{k}{Len(\Phi (\mathbb{V}_{n}))},1+\frac{n-k}{Len(\Phi (\mathbb{V}_{n}))}\} &{\rm by}\; (\ref{OPTUBINCASE2})\\
    &\leq \min\{\frac{k}{1+(k-1)\theta},1+\frac{n-k}{1+(k-1)\theta}\} &{\rm by}\;(\ref{SOAoveralllength})\\
    &\leq \frac{k}{1+(k-1)  \theta}= \frac{\sqrt{9  k^2-14k+9}-k-1}{2(k-1)}+1&{\rm by}\;\theta=\frac{\sqrt{9  k^2-14k+9}-k-1}{4(k-1)}
\end{flalign*}\end{small}

By Case 1 and Case 2, the proof completes.
\end{proof}
\begin{algorithm}[tb] 
\caption{\quad \textbf{SOA$\bm{_{\rm AN}}$} }
\label{soaan}
The SOA${_{\rm AN}}$ remains the same as the Algorithm \ref{FixHitching} by discarding  the \textbf{else if} branch of the quota-enough condition in Lines 6-7 and setting $\theta=\frac{\sqrt{9 k^2-14k+9}-k-1}{4(k-1)}$;
\end{algorithm}
\begin{corollary}\label{ub_ulan}
For UL-AN, SOA$_{\rm AN}$ runs in $O(n)$ time and achieves a competitive ratio no larger than
$\frac{\sqrt{9 k^2-14k+9}-k-1}{2(k-1)}+1$ for any limited time frame.
\end{corollary}
\begin{proof}
In the UL-AN setting, where the number $|\mathbb{V}|$ of total released sub-intervals is unknown to the online algorithm in advance, we discuss two cases in any limited time period $T$ for acceptance.
\\
\textbf{Case 1.} SOA$_{\rm AN}$ runs out of its quota $k$ in time frame $T$ and accepts $k$ sub-intervals by the single threshold $\theta=\frac{\sqrt{9 k^2-14k+9}-k-1}{4(k-1)}$. This implies the overall covered length by SOA$_{\rm AN}$ as
$Len(\Phi(\mathbb{V}))\geq 1+(k-1)\theta$. On the other hand, we have a trivial upper bound on the covered length of OPT as $Len(OPT)\leq k$. Hence, we have the ratio 
\begin{equation*} \label{corollary4case1}
    \rho\leq \frac{k}{1+(k-1)\theta}=\frac{\sqrt{9 k^2-14k+9}-k-1}{2(k-1)}+1.
\end{equation*}
\\
\textbf{Case 2.} SOA$_{\rm AN}$ still remains quota ($\geq 1$) after the time frame $T$. We get the ratio as $\rho\leq 1+2\theta=\frac{\sqrt{9 k^2-14k+9}-k-1}{2(k-1)}+1$ by a similar proof idea as in the Case 1 of Theorem 4.

By Cases 1-2, we get the upper bound of SOA$_{\rm AN}$ as $\frac{\sqrt{9 k^2-14k+9}-k-1}{2(k-1)}+1$.
\end{proof}
The SOA algorithm can solve the the flexible-length case. Using a similar analysis idea as in Theorem \ref{FixHitchingratio}, we have the following Theorem \ref{unit-sumUB1} and Theorem \ref{Varied_FixHitchingratio}.
\begin{theorem} \label{unit-sumUB1}
For US-UN, SOA runs in $O(n)$ time and achieves a competitive ratio no larger than
$\min\{\frac{\sqrt{1+2(k-1)(n-k)}-1}{k-1}+1,\frac{\sqrt{9  k^2-14k+9}-k-1}{2(k-1)}+1\}$.
\end{theorem}
\begin{proof}
SOA runs in $O(n)$ time since it runs in no more than $n$ iterations in which each runs in $O(1)$ time. In the US-UN setting, the sum of length of  sub-intervals in each $V_i\in\mathbb{V}$ is a unit. We also discuss two cases. 
\\
\textbf{Case 1}. SOA accepts $V_n$ (the last candidate released). Suppose, without loss of generality, that $\bigcup\limits_{V\in \Phi (\mathbb{V}_{i-1})}V$ consists of a number \textbf{$x$} of disjoint interval, denoted by $\mathbb{A}_1$, $\mathbb{A}_2$,..., $\mathbb{A}_x$ respectively. Note that each rejected candidate can contribute an additional length at most $\theta$ to $Len(\Phi (\mathbb{V}_{i-1}))$, as otherwise it would have been accepted. On one hand, $ Len(OPT)\leq Len(\Phi (\mathbb{V}_{n}))+ \theta   (n-k)$ holds naturally since SOA totally rejects ($n-k$) candidates in $\mathbb{V}_{i-1}$. On the other hand, $  Len(OPT)\leq Len(\Phi (\mathbb{V}_{n}))+ 2 \theta   x$ because there are $x$ disjoint intervals formed by the candidates accepted by SOA. In summary, 
\begin{equation}\label{OPTUBCASE1fixed}
    Len(OPT)\leq Len(\Phi (\mathbb{V}_{n}))+ \min\{2 \theta   x,\theta   (n-k)\}
\end{equation}
Following the proof idea of the Case 1 in Theorem \ref{FixHitchingratio}, we also get 
$$\rho\leq 1+2\theta\leq \min\{\frac{\sqrt{1+2(k-1)(n-k)}-1}{k-1}+1,\frac{\sqrt{9  k^2-14k+9}-k-1}{2(k-1)}+1\}$$
\\
\textbf{Case 2}.  SOA does not accept $V_n$. This implies the quota-enough condition is not triggered during the execution and SOA accepts $k$ candidates by the threshold-accepting condition. One one hand, we have $Len(SOA)\geq 1+(k-1)\theta$ as each accepted candidate by SOA (excluding the first accepted one, which actually contributes a length of $1$) contributes at least $\theta$ to the overall covered length of the target interval $[0,a]$ in the worst case.  On the other hand, it naturally holds that $Len(OPT)\leq \min\{k,Len(\Phi(\mathbb{V}_n))+(n-k)\}$ as a similar analysis in Theorem 5.
By Proposition 2 and Theorem 4, we have
\begin{small}
\begin{equation*}
\begin{split}
    \rho&=\frac{Len(OPT)}{Len(\Phi (\mathbb{V}_{n}))}\\
   &\leq \min\{\frac{\sqrt{1+2(k-1)(n-k)}-1}{k-1}+1,\frac{\sqrt{9  k^2-14k+9}-k-1}{2(k-1)}+1\}
\end{split}
\end{equation*}
\end{small}
By Case 1 and Case 2, the proof completes.
\end{proof}
\begin{theorem} \label{Varied_FixHitchingratio}
For FL-UN, SOA runs in $O(n)$ time and achieves a competitive ratio no larger than
$\min\{\frac{\sqrt{1+2(k-1)(n-k)m}-1}{k-1}+1,\frac{\sqrt{(1+8m)  k^2-(6+8m)k+9}-k-1}{2(k-1)}+1\}$\\
in which $m$ indicates the maximum possible length of a sub-interval.
\end{theorem}
\begin{proof}
The running time of SOA is proved to be $O(n)$ in Theorem 4.

In the FL-UN setting, the length of each sub-interval varies in a fixed range $[1,m]$. Given the value $m$ of the maximum length of a sub-interval  the total number $n$ of potential sub-intervals and the quota $k$, we take $$
    \theta=\min\{\frac{\sqrt{1+2(k-1)(n-k)m}-1}{2k-2},\frac{\sqrt{(1+8m)  k^2-(6+8m)k+9}-k-1}{4(k-1)}\}$$ Notice that
$\frac{\sqrt{1+2(k-1)(n-k)m}-1}{2k-2}$ decreases in $k$, while $\frac{\sqrt{(1+8m)  k^2-(6+8m)k+9}-k-1}{4(k-1)}$ increases in $k$, when $k\in\{2,3,\cdots, n-1\}$.  We show the competitive ratio by the following two cases. 
\\ \textbf{Case 1.}  SOA accepts $V_n$ (the last sub-interval released). Following the proof idea of Case 1 in Theorem  \ref{FixHitchingratio}, we have $\rho\leq 1+2\theta\leq \min\{\frac{\sqrt{1+2(k-1)(n-k)m}-1}{k-1}+1,\frac{\sqrt{(1+8m)  k^2-(6+8m)k+9}-k-1}{2(k-1)}+1\}$
\\
\textbf{Case 2.  }SOA does not accept $V_n$. This means the quota-enough condition is not triggered during the execution and SOA accepts $k$ sub-intervals by the threshold-accepting condition. Suppose that $V_i\in\{V_k,...,V_{n-1}\}$ is the last accepted sub-interval by SOA, i.e., $|\Phi(\mathbb{V}_i)|=k$ and
$Len(\Phi(\mathbb{V}_n))=Len(\Phi(\mathbb{V}_{i}))$. In other words, SOA misses all sub-intervals in $\{V_{i+1},...,V_n\}$ which can be accepted by OPT. Since the algorithm can miss at most $n-k$ sub-intervals, OPT can get an overall length less than $Len(\Phi(\mathbb{V}_n))+(n-k)m$. Also, OPT cannot get an overall length over $km$ by its quota $k$ in this unit-length case. \begin{equation}\label{OPTUBINCASE20}
    Len(OPT)\leq \min\{km,Len(\Phi(\mathbb{V}_n))+(n-k)m\}
\end{equation}
\\ Following the proof idea of Case 2 in Theorem 4, we have
\begin{small}
\begin{equation*}
\begin{split}
    &\rho=\frac{Len(OPT)}{Len(\Phi (\mathbb{V}_{n}))}\\
    &\leq \frac{\min\{km,1+(k-1)  \theta+(n-k)m\}}{1+(k-1)  \theta}\\
    & \leq \min\{\frac{\sqrt{1+2(k-1)(n-k)m}-1}{k-1}+1,\frac{\sqrt{(1+8m)  k^2-(6+8m)k+9}-k-1}{2(k-1)}+1\}
\end{split}
\end{equation*}\end{small}

By Case 1 and Case 2, the proof completes.
\end{proof}
\begin{corollary}\label{ub_flan}
For FL-AN, SOA$_{\rm AN}$ runs in $O(n)$ time and achieves a competitive ratio no larger than $\frac{\sqrt{(1+8m)  k^2-(6+8m)k+9}-k-1}{2(k-1)}+1$ for any limited accepting time frame, in which $m$ indicates the maximum possible length of a sub-interval.
\end{corollary} 
\begin{proof}
Since the $n=|\mathbb{V}|$ is discarded in the FL-AN setting, we can upper bound the overall covered length by OPT as $km$. Further, by a similar proof idea   as in Theorem 6, one can derive the upper bound of this corollary. \end{proof}
\subsection{Double-threshold Online Algorithm}\label{doasection}
Built upon SOA, we now present the Double-threshold Online Algorithm (\textbf{DOA}) under the UN setting, 
which remains the same as Algorithm \ref{FixHitching} but extends the single threshold $\theta$ in the  threshold-meeting condition to two thresholds $\theta_1$ and $\theta_2$ (by using $\theta_1$ for exploration and $\theta_2$ for exploitation). Specifically, SOA changes the threshold from $\theta_1$ to $\theta_2$ once accepting $\omega$ sub-intervals, in which the values of $(\omega ,\theta_1,\theta_2)$ are given later by solving the non-linear program (\romannumeral1-\romannumeral8). Before that, we first give the competitive analysis of DOA. 

Denote $j$ as the number of disjoint intervals formed by the sub-intervals accepted by DOA. When DOA accepts less than $k$ sub-intervals by threshold, the overall length achieved by OPT is no more than $Len(\Phi(\mathbb{V}_n))+2j \theta$ and certainly no more than the quota $k$, implying Lemma \ref{secondlemmaindouble}. When DOA accepts $k$ sub-intervals by threshold, the overall length of OPT should be no more than ($n-k+1+(\omega -1)\theta_1+(k-\omega)\theta_2$) and also no more than $k$, implying Lemma \ref{thirdlemmaindouble}.
\begin{lemma}\label{secondlemmaindouble}
In UL-UN, when DOA accepts $i$ ($1\leq i\leq k-1$) sub-intervals by threshold and quota-enough accepts $k-i$ sub-intervals, OPT can achieve an overall length at most $\frac{\min\{k,j+(i-j)\theta_1+2  j\theta_1\}}{j+(i-j)\theta_1}$ times of DOA length for $1\leq i\leq \omega -1$ or at most  $\frac{\min\{k,j+(\omega-1)\theta_1+(i-\omega+2j)\theta_2\}}{j+(\omega-1)\theta_1+(i-\omega)\theta_2}$ times  of DOA length for $\omega \leq i\leq k-1$.
\end{lemma}
\begin{lemma}\label{thirdlemmaindouble}
In UL-UN, when DOA threshold-accepts $k$ sub-intervals, OPT can achieve an overall length at most $\frac{\min\{k,n-k+1+(\omega-1)\theta_1+(k-\omega)\theta_2\}}{1+(\omega-1)\theta_1+(k-\omega)\theta_2}$ times of DOA's.
\end{lemma}
\begin{theorem}\label{competitiveratioindouble}
In UL-UN, the competitive ratio of DOA is upper bounded by
\begin{small}
\begin{equation}\label{ratiocandidatedouble}
C=\max\{1+2\theta_1, 1+\frac{2\theta_2}{1+\frac{\omega-s}{s}\theta_1}, \frac{k}{s+1+(\omega-s-1) \theta_1}, \frac{\min\{k,n-k+q\}}{q}\}
\end{equation}\end{small}
\\where $q=1+(\omega-1)\theta_1+(k-\omega)\theta_2$, $s=\frac{k+(1-\omega)\theta_1-2\theta_2}{1+2\theta_2-\theta_1}$.
\end{theorem}
\begin{proof}
Suppose, w.l.o.g., that DOA threshold-accepts $i$ sub-intervals and quota-enough accepts the other ($k-i$) sub-intervals. Let $j$ denote the number of disjoint intervals formed by the $k$ accepted sub-intervals. In the worst case scenario, the $k-i$ quota-enough accepted sub-intervals only contribute an additional length of zero to the covered length of the interval $[0,a]$ by previously threshold-accepted sub-intervals. We discuss in three cases,
\\
\textbf{Case 1.} $1\leq i\leq \omega-1$. By Lemma \ref{secondlemmaindouble}, we have the ratio 

$\rho \leq \frac{\min\{k,Len(\Phi(\mathbb{V}_n))+2  j\theta_1\}}{Len(\Phi(\mathbb{V}_n))}\leq 
\frac{\min\{k,j+(i-j)\theta_1+2  j\theta_1\}}{j+(i-j)\theta_1}\leq 1+2\theta_1$.
\\
\textbf{Case 2.} $\omega\leq i\leq k-1$. Denote $s+1$ as the minimum number of disjoint intervals formed by the accepted sub-intervals in an optimal solution. Suppose w.l.o.g., OPT achieves its maximum overall length $k$ in the worst-case scenario. We have  $s+1+(\omega-s)\theta_1+2(s+1)\theta_2 \geq k
$ and $s+(\omega-s)\theta_1+2s\theta_2 < k$ in the worst case, 
\begin{center}\label{solutionofs}
    $\frac{k+(1-\omega)\theta_1-2\theta_2}{1+2\theta_2-\theta_1}\leq s\leq \frac{k-\omega\theta_1}{1-\theta_1+2\theta_2}$
\end{center}
Further, we get $s=\frac{k+(1-\omega)\theta_1-2\theta_2}{1+2\theta_2-\theta_1}\in [1,\omega-1]$ satisfying Inequality (\ref{solutionofs}).\\
\textbf{Case 2.1.} $j\leq s$. By Lemma \ref{secondlemmaindouble},  the ratio is upper bounded by \begin{center}
$\frac{j+(\omega-j)\theta_1+(i-\omega)\theta_2+2j\theta_2}{j+(\omega-j)\theta_1+(i-\omega)\theta_2}\leq \frac{s+(\omega-s)\theta_1+2s\theta_2}{s+(\omega-s)\theta_1}=1+\frac{2\theta_2}{1+\frac{\omega-s}{s}\theta_1}
$\end{center} 
\textbf{Case 2.2.} $j\geq s+1$. By Lemma \ref{secondlemmaindouble}, the ratio is upper bounded by \begin{center}
$\frac{k}{j+(\omega-j)\theta_1+(i-j)\theta_2}\leq
\frac{k}{s+1+(\omega-s-1) \theta_1}$
\end{center} 
\textbf{Case 3.} $i=k$. By Lemma \ref{thirdlemmaindouble}, the competitive ratio of DOA is upper bounded by $\frac{\min\{k,n-k+1+(\omega-1)\theta_1+(k-\omega)\theta_2\}}{1+(\omega-1)\theta_1+(k-\omega)\theta_2}$.

By Cases 1, 2, 3, competitive ratio of DOA is upper bounded by  (\ref{ratiocandidatedouble}).
\end{proof}
To find the timing $\omega$ and the thresholds ($\theta_1,\theta_2$) that optimize the competitive ratio of DOA, we propose the following nonlinear program to minimize the maximum value (denoted by $C$) of the competitive ratio in Equation (\ref{ratiocandidatedouble}), where constraints (\romannumeral2)-(\romannumeral6) are transformed from Equation (\ref{ratiocandidatedouble}) respectively and constraints (\romannumeral7) \footnote{Constraint (\romannumeral7) actually can be  restricted, by calculation, to $\left \lceil \frac{k+1}{5} \right \rceil\leq \omega \leq k$.} and  (\romannumeral8) are naturally required.
\begin{small}
\begin{alignat}{2}\label{doublethresholdsolution}
\min\limits_{(\omega,\theta_1,\theta_2)}\quad & C & \tag{\romannumeral1} \\
\mbox{s.t.}\quad
& C\geq 1+2\theta_1, \quad& \tag{\romannumeral2}\\
& C\geq 1+\frac{2\theta_2}{1+\frac{\omega-s}{s}\theta_1} &  \tag{\romannumeral3}\\
& C\geq \frac{k}{s+1+(\omega-s-1) \theta_1}& \tag{\romannumeral4}\\
& C\geq \frac{\min\{k,n-k+1+(\omega-1)\theta_1+(k-\omega)\theta_2\}}{1+(\omega-1)\theta_1+(k-\omega)\theta_2} & \tag{\romannumeral5}\\
& s=\frac{k+(1-\omega)\theta_1-2\theta_2}{1+2\theta_2-\theta_1}& \tag{\romannumeral6}\\
& 1 \leq \omega \leq  k&\tag{\romannumeral7}\\
 &0<\theta_1<\theta_2\leq 1, & \tag{\romannumeral8}
\end{alignat}\end{small}
\begin{theorem}\label{doublehitchingcompetitveratio}
DOA with ($\omega,\theta_1,\theta_2$) returned by program (\romannumeral1)-(\romannumeral8) achieves the best worst-case performance of online algorithms with two thresholds.
\end{theorem}
\begin{figure}
    \centering
    \includegraphics[width=6cm]{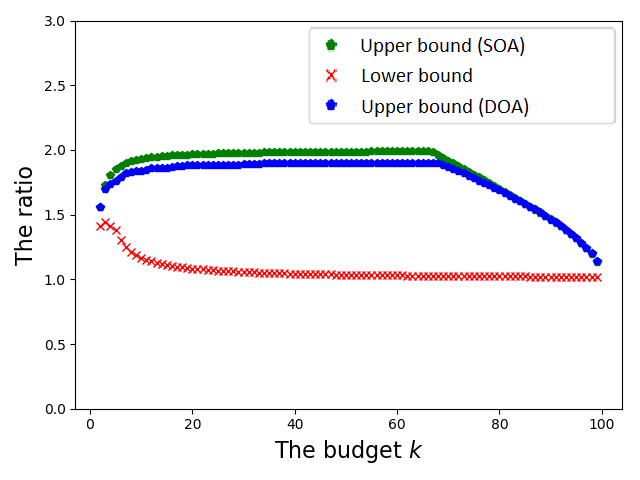}
    \caption{Performance among DOA, SOA and the lower bound in UL-UN.}
    \label{comparisonwithdoublehitching}
\end{figure}
Since the program (\romannumeral1-\romannumeral8) is nonlinear and is complicated when transformed into a linear programming,  we search its approximated solution under the UL-UN setting by giving the precision of $\theta$ as 0.01 and $n=100$. According to the searching result, we observe that the $\omega$ value should be set at around $0.8k$ and $\theta_1<\theta_2$. The double-threshold algorithm DOA improves the performance of the single-threshold based algorithm (see Figure \ref{comparisonwithdoublehitching} below). What is worth noting is that, when the ratio $\frac{k}{n}$ of the quota over the total number of online sub-intervals is relatively small (\textit{resp.} large), we find that more quota induces worse (\textit{resp}. better) performances of both SOA and DOA since OPT has more chances to gain values from those missed sub-intervals by our algorithms (\textit{resp}. since online algorithms have fewer chances to miss values from OPT). The turning point of $\frac{k}{n}$ is around $\frac{2}{3}$ in SOA since the two items of the competitive ratio in Theorem \ref{FixHitchingratio} are monotone decreasing and increasing, respectively, with regard to $k$, and meet when $\frac{k}{n}\approx \frac{2}{3}$. Interestingly, the turning point of $\frac{k}{n}$ is also around $\frac{2}{3}$ in DOA, see the example in Figure \ref{comparisonwithdoublehitching}.  We can also extend to more than two thresholds, yet the analysis will be more involved with only mild improvement. Particularly, when the thresholds in an algorithm are non-increasing as accepting sub-intervals, we have the following theorem.

\begin{theorem}\label{multithresholdstheorem}
SOA outperforms any online deterministic algorithm that accepts sub-intervals by non-increasing thresholds.
\end{theorem}
\begin{proof}
We discuss the competitive ratio of such ALG in \textit{the UL setting }in a limited time frame $T$ in this proof. First, we define the online algorithm (ALG) concerned in this theorem as follows.

~\\
\textsc{An online algorithm with non-increasing thresholds}: uses threshold $\theta_j$ for the $i$th sub-interval to be accepted, satisfying that $\theta_1\geq \theta_2\cdots\geq \theta_k$. Specifically, suppose the ALG accepts a set of $i$ sub-intervals already, then, ALG accepts a new sub-interval if and only if the new sub-interval contributes an additional length to ALG equal to or larger than $\theta_{i+1}$. The algorithm stops only when either it runs out of the quota $k$ or the time frame concerned ends.

~\\
For ease of understanding, we recall the following notations used in the proof.
\begin{itemize}
    \item $Len(\Phi(\mathbb{V}))$ denote the overall covered length of the target interval $[0,a]$ by accepted sub-intervals in ALG in the given time period $T$;
    \item $Len(OPT)$ denotes the overall covered length of the target $[0,a]$ by accepted sub-intervals in OPT in the time frame $T$;
\end{itemize}
Suppose, w.l.o.g., that ALG finally accepts $y$ sub-intervals within the given time frame $T$, formulating a number $x$ of disjoint intervals of the target $[0,a]$. Now, we discuss $k$ cases with regard to $y$, in each of which we further discuss a number $y$ of subcases with regard to $x$. 
\\\textbf{Case 1.} $y=1$. Clearly, we have $x=1$. Note that $Len(OPT)$ can not exceed over $Len(\Phi(\mathbb{V}))$ by $2\theta_2$ as otherwise ALG can accept another sub-interval by $\theta_2\leq\theta_1$, we have $\rho_{\rm case\;1.1}= \frac{1+2\theta_2}{1}=1+2\theta_2$;
\\
\textbf{Case 2.} $y=2$. We have two subcases.
\begin{itemize}
    \item\textbf{Case 2.1.} $x=1$. We have $Len(\Phi(\mathbb{V}))\geq 1+\theta_2$ and $Len(OPT)\leq 1+\theta_2+2\theta_3$ as otherwise ALG can accept the third sub-interval by threshold $\theta_3\leq \theta_2$. Hence, $\rho_{\rm case\;2.1}= \frac{Len(OPT)}{Len(\Phi(\mathbb{V}))}=\frac{1+\theta_2+2\theta_3}{1+\theta_2}=1+\frac{2\theta_3}{1+\theta_2}$; 
    \item \textbf{Case 2.2.} $x=2$. We have $Len(\Phi(\mathbb{V}))=2$ in this unit-length setting and $Len(OPT)\leq 2+4\theta_3$ as otherwise ALG can accept the third sub-interval by threshold $\theta_3\leq \theta_2$. Hence, $\rho_{\rm case\;2.2}= \frac{Len(OPT)}{Len(\Phi(\mathbb{V}))}\leq \frac{2+4\theta_3}{2}=1+2\theta_3$ which is larger than $\rho_{\rm case\;2.1}=1+\frac{2\theta_3}{1+\theta_2}$.
\end{itemize}
overall, the ratio of Case 2 is upper bounded by $\rho_{\rm case\;2.2}=1+2\theta_3<1+\theta_2$

$\;\bm{\vdots}$
\\
\textbf{Case} \bm{$[\frac{k}{2}]$}. $y=[\frac{k}{2}]$. We have $[\frac{k}{2}]$ subcases.
\begin{itemize}
\item\textbf{Case $[\frac{k}{2}]$.1.} $x=1$. We have $Len(\Phi(\mathbb{V}))\geq 1+\sum_{i=2}^{[\frac{k}{2}]}\theta_i$ and
     \begin{equation*}
         Len(OPT)\leq Len(\Phi(\mathbb{V}))+2\theta_{[\frac{k}{2}]+1}=1+\sum_{i=2}^{[\frac{k}{2}]}\theta_i+2\theta_{[\frac{k}{2}]+1}
    \end{equation*}
         in which the inequality holds as otherwise ALG can accept the ($[\frac{k}{2}]+1$)th sub-interval by threshold $\theta_{[\frac{k}{2}]+1}\leq \theta_{[\frac{k}{2}]}$. Hence, 
         \begin{equation*}
         \rho_{{\rm case\;}[\frac{k}{2}].1.}= \frac{Len(OPT)}{Len(\Phi(\mathbb{V}))}\leq 1+\frac{2\theta_{[\frac{k}{2}]+1}}{1+\sum_{i=2}^{[\frac{k}{2}]}}\theta_i
         \end{equation*}
 \item\textbf{Case $[\frac{k}{2}]$.2.} $x=2$. We have $Len(\Phi(\mathbb{V}))\geq 2+\sum_{i=3}^{[\frac{k}{2}]}\theta_i$ and
     \begin{equation*}
         Len(OPT)\leq Len(\Phi(\mathbb{V}))+4\theta_{[\frac{k}{2}]+1}=2+\sum_{i=3}^{[\frac{k}{2}]}\theta_i+4\theta_{[\frac{k}{2}]+1}
    \end{equation*}
         in which the inequality holds as otherwise ALG can accept the ($[\frac{k}{2}]+1$)th sub-interval by threshold $\theta_{[\frac{k}{2}]+1}\leq \theta_{[\frac{k}{2}]}$. Hence, 
         \begin{equation*}
         \rho_{{\rm case\;}[\frac{k}{2}].2.}= \frac{Len(OPT)}{Len(\Phi(\mathbb{V}))}\leq 1+\frac{4\theta_{[\frac{k}{2}]+1}}{2+\sum_{i=3}^{[\frac{k}{2}]}\theta_i}=1+\frac{2\theta_{[\frac{k}{2}]+1}}{1+\frac{\sum_{i=3}^{[\frac{k}{2}]}\theta_i}{2}}
         \end{equation*}
         
           \quad$\;\bm{\vdots}$
 \item\textbf{Case $[\frac{k}{2}]$.j-1.} $x=j-1$\footnote{Here, the $j$ is chosen such that $2j+2\leq k<3j$.}. We have $Len(\Phi(\mathbb{V}))\geq j-1+\sum_{i=j}^{[\frac{k}{2}]}\theta_i$ and
     \begin{equation*}
         Len(OPT)\leq Len(\Phi(\mathbb{V}))+2(j-1)\theta_{[\frac{k}{2}]+1}
    \end{equation*}
         in which the inequality holds as otherwise ALG can accept the ($[\frac{k}{2}]+1$)th sub-interval by threshold $\theta_{[\frac{k}{2}]+1}\leq \theta_{[\frac{k}{2}]}$. Hence, 
         \begin{equation*}
         \rho_{{\rm case\;}[\frac{k}{2}].j-1.}= \frac{Len(OPT)}{Len(\Phi(\mathbb{V}))}\leq 1+\frac{2(j-1)\theta_{[\frac{k}{2}]+1}}{j-1+\sum_{i=j}^{[\frac{k}{2}]}\theta_i}=1+\frac{2\theta_{[\frac{k}{2}]+1}}{1+\frac{\sum_{i=3}^{[\frac{k}{2}]}\theta_i}{j-1}}
         \end{equation*} 
         note that $\rho_{{\rm case\;}[\frac{k}{2}].j-1.}\geq \rho_{{\rm case\;}[\frac{k}{2}].1.}=1+\frac{2\theta_{[\frac{k}{2}]+1}}{1+\sum_{i=3}^{[\frac{k}{2}]}\theta_i}$.
\item\textbf{Case $[\frac{k}{2}]$.j.} $x=j$. We have $Len(\Phi(\mathbb{V}))\geq j+\sum_{i=j+1}^{[\frac{k}{2}]}\theta_i$ and
     \begin{equation*}
         Len(OPT)\leq Len(\Phi(\mathbb{V}))+(k-j)\theta_{[\frac{k}{2}]+1}
    \end{equation*}
         in which the inequality holds as OPT can exceed over ALG by at most $(k-j)\theta_{[\frac{k}{2}]+1}$ by $k<3j$ and the budget $k$. Hence, 
         \begin{equation*}
         \rho_{{\rm case\;}[\frac{k}{2}].j.}= \frac{Len(OPT)}{Len(\Phi(\mathbb{V}))}\leq 1+\frac{(k-j)\theta_{[\frac{k}{2}]+1}}{j+\sum_{i=j+1}^{[\frac{k}{2}]}\theta_i}= 1+\frac{2\theta_{[\frac{k}{2}]+1}}{\frac{2j}{k-j}+\frac{2\sum_{i=j+1}^{[\frac{k}{2}]}\theta_i}{k-j}}
         \end{equation*} 
         note that $\rho_{{\rm case\;}[\frac{k}{2}].j.}\leq \rho_{{\rm case\;}[\frac{k}{2}].1.}=1+\frac{2\theta_{[\frac{k}{2}]+1}}{1+\sum_{i=3}^{[\frac{k}{2}]}\theta_i}\leq \rho_{{\rm case\;}[\frac{k}{2}].j-1.}$, in which the first inequality holds by $k\leq 3k$.
\item\textbf{Case $[\frac{k}{2}]$.j+1.} $x=j+1$. We have $Len(\Phi(\mathbb{V}))\geq j+1+\sum_{i=j+2}^{[\frac{k}{2}]}\theta_i$ and
     \begin{equation*}
         Len(OPT)\leq Len(\Phi(\mathbb{V}))+k\theta_{[\frac{k}{2}]+1}
    \end{equation*}
         in which the inequality holds as OPT can exceed over ALG by at most $k\theta_{[\frac{k}{2}]+1}$ since $k<3j$. Hence, 
         \begin{equation*}
         \rho_{{\rm case\;}[\frac{k}{2}].j+1.}= \frac{Len(OPT)}{Len(\Phi(\mathbb{V}))}\leq 1+\frac{k\theta_{[\frac{k}{2}]+1}}{j+1+\sum_{i=j+2}^{[\frac{k}{2}]}\theta_i}
         \end{equation*} 
        note that $\rho_{{\rm case\;}[\frac{k}{2}].j+1.}\leq \rho_{{\rm case\;}[\frac{k}{2}].j.}$ as $\theta\leq1$.
        
         \quad$\;\bm{\vdots}$
\end{itemize}
 We note that the ratios subcases of this Case $[\frac{k}{2}]$ increase as the $x$ increases within $\{1,\cdots,j-1\}$ and  further decrease as the  $x$ increases within $\{j,\cdots,y\}$.  Overall, we have the ratio of this case as
 \begin{equation*}
 \begin{split}
      \rho_{\rm case [\frac{k}{2}]}&=\max\limits_{1 \leq x\leq [\frac{k}{2}] }\{\rho_{\rm case [\frac{k}{2}].x}\}=\rho_{\rm case [\frac{k}{2}].j-1}\\
      &= 1+\frac{2\theta_{[\frac{k}{2}]+1}}{1+\frac{\sum_{i=3}^{[\frac{k}{2}]}\theta_i}{j-1}}\\
      &<1+2\theta_{[\frac{k}{2}]+1}\\
      &\leq \rho_{\rm case 1}=1+2\theta_2. \end{split}
  \end{equation*}

$\;\bm{\vdots}$
\\
\textbf{Case \bm{$k$}.} $y=k$. Clearly, $Len(OPT)\leq k$ while $Len(\Phi(\mathbb{V}))\geq 1+\theta_1+\cdots+\theta_k$. Hence, we have $\rho_{\rm case \;k}\leq \frac{k}{1+\sum_{i=2}^{k}\theta_i}$

~\\
Notice the the ratios of the cases decreases as the $y$ increase till $y=k-1$. Therefore, the competitive ratio of ALG is upper bounded by 
\begin{align*}
&\max\{\frac{k}{1+\sum_{i=2}^{k}\theta_i}, 1+2\theta_2\}&\\
&\geq \max\{\frac{k}{1+(k-1)\theta_2}, 1+2\theta_2\}&\\
&\geq \frac{\sqrt{9 k^2-14k+9}-k-1}{2(k-1)}+1& {\rm ratio \;of \;SOA\;in\;
AN}\\
&\geq \min\{\frac{\sqrt{1+2(k-1)(n-k)}-1}{k-1}+1,\frac{\sqrt{9 k^2-14k+9}-k-1}{2(k-1)}+1\}&{\rm ratio \;of \;SOA\;in \;UN}\\
\end{align*}
in which the first inequality holds by $\theta_1\geq \theta_2\cdots\geq \theta_k$, and the second inequality holds by $\frac{\sqrt{9 k^2-14k+9}-k-1}{4(k-1)}=\min\limits_{\theta_2}\max\{\frac{k}{1+(k-1)\theta_2}, 1+2\theta_2\}$.

The proof completes.
\end{proof}

\section{Concluding Remarks}
This paper studies the online maximum $k$-coverage problem on a line without preemption. With regard to the length of each sub-interval and the number of totally released sub-intervals, we comprehensively consider different settings in this paper. Our contribution is three-fold. 

\textbf{First}, we present lower bounds on the competitive ratio for the settings respectively. \textbf{Second}, we propose an optimal solution for the offline problem where the sequence of offline sub-intervals is given to the decision-maker at the very beginning. \textbf{Third}, we present two online algorithms, including a single-threshold-based algorithm SOA and a double-threshold-based algorithm (DOA). DOA uses its first threshold (which is usually set below 0.5) for exploration in accepting the first $\left [0.8 k\right ]$ released sub-intervals and its second threshold  (which is set larger than the first threshold) for exploitation in accepting the last $k-\left [0.8 k\right ]$ sub-intervals. We prove that SOA achieves competitive ratios close to the lower bounds, respectively, and DOA, with its parameters computed by our proposed program, improves the performance of SOA slightly. In addition, we show that any online deterministic algorithm that accepts sub-intervals by non-increasing thresholds, cannot achieve a competitive ratio better than SOA no matter how many thresholds  the algorithm uses.

For the future work, we may consider the case that different sub-intervals are associated with different costs instead of unit costs in this paper, considering that candidates may call for different payments in crowding-sourcing activities.

\textbf{Acknowledgements.} This work was done when Songhua Li was visiting the Singapore University of Technology and Design. Minming Li is also from City University of Hong Kong Shenzhen Research Institute, Shenzhen, P.R. China. The work described in this paper was partially supported by Project 11771365 supported by NSFC.

\newpage
\appendix

\section{Dynamic Programming Solution to The Unit-length Setting} \label{app_DP_unit}
Both the unit-length case and the flexible-length case can be regarded as a special case of the arbitrary-length case. Hence, our dynamic programming-based offline solution can be easily applied to the other two cases. Here, we take the unit-length case as an example.

To solve the offline problem, we first sort offline sub-intervals in $\mathbb{V}$ in non-decreasing order of their end locations, which runs in $O(kn+n\log n)$ time in the worst case. In the following solution, we abuse notation to denote $\{V_1,V_2,...,V_n\}$ as the sub-interval set we got after sorting, i.e., $d_1\leq d_2 \leq \cdots \leq d_n$, and hence their start locations satisfy $o_1\leq o_2 \leq \cdots \leq o_n$. Suppose, without loss of generality, that the optimal offline solution (i.e., OPT) accepts sub-intervals in $\mathbb{V}_n$ in decreasing order of the subscripts. Note that the overall length of OPT would not decrease when replacing the right-most sub-interval by $V_n$, implying the following Observations 1 \& 2. 
\\
\textbf{Observation 1}.
In an optimal offline solution, $V_n$ is accepted.
\\
\textbf{Observation 2}.
Once $V_i\in \mathbb{V}$ is accepted by OPT, one of the following two sub-intervals is accepted by OPT.
\begin{equation}\label{unitcandidateone}
    V_{\lambda(i)}=\mathop{\arg\max}\limits_{\{V_j\in \mathbb{V}|d_i-1\leq d_j\leq d_i\}}\{|d_i-d_j|\}
\end{equation}
or 
\begin{equation}\label{unitcandidatetwo}
 V_{\mu (i)}=\mathop{\arg\min}\limits_{\{V_j\in \mathbb{V}|d_i-d_j\geq 1\}}\{|d_i-d_j|\} 
\end{equation} 
in which $\lambda(i)$ and $\mu(i)$ indicate the corresponding subscripts of the sub-intervals in $\mathbb{V}$ respectively.

To see why, we give an example in Figure \ref{DPobservations}, in which $V_{\lambda(i)}=V_{i-2}$ since $V_{i-2}$ contributes an additional length to OPT more than $V_{i-1}$ does, and $ V_{\mu (i)}=V_{i-3}$ since $V_{i-3}$ contributes an additional length to OPT no less than $V_{i-4}$ or $V_{i-5}$ does.
\begin{figure}
    \centering
    \includegraphics[width=8cm]{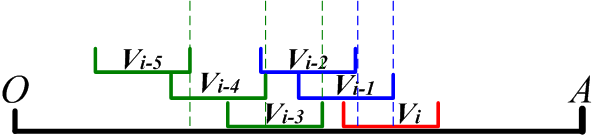}
    \caption{An explanation of Observation 2}
    \label{DPobservations}
\end{figure}

\textbf{Dynamic programming approach.} Given the offline sub-intervals in $\mathbb{V}_n$, the optimal solution $\chi (\mathbb{V}_n,k)$ can be obtained via the following equations by Observations 1-2: when both $ V_{\lambda(n)}$ and $ V_{\mu (n)}$ exist, we have Equation (\ref{baseequationone}); when $ V_{\mu (n)}$ exists but $V_{\lambda(n)}$ does not, we have Equation (\ref{baseequationotwo}); when $V_{\lambda(n)}$ exists but $ V_{\mu (n)}$ does not, we have Equation (\ref{baseequationothree}); when neither $V_{\lambda(n)}$ nor $ V_{\mu (n)}$ exists, we have Equation (\ref{baseequationofour}). For base cases, we have Equation (\ref{base01}) when $k=0$ for arbitrary set $\mathbb{V}_i$, and Equation (\ref{base02}) when $k=1$ for arbitrary set $\mathbb{V}_i$. 

Since our dynamic programming approach  generates  $O(kn)$ different intermediate states in which each state is calculated in $O(1)$ time, our offline optimal solution totally runs in $O(kn+n\log n)$ time, including the preliminary sorting step. We have
\begin{equation}\label{baseequationone}
  \chi (\mathbb{V}_n,k)=\max\{\chi (\mathbb{V}_{\mu (n)},k-1)+1,\chi (\mathbb{V}_{\lambda(n)},k-1)+1-\Lambda(V_n,V_{\lambda (n)})\}   
\end{equation}
\begin{equation}\label{baseequationotwo}
   \chi (\mathbb{V}_n,k)=\chi (\mathbb{V}_{\mu (n)},k-1)+1,
\end{equation}
\begin{equation}\label{baseequationothree}
   \chi (\mathbb{V}_n,k)=\chi (\mathbb{V}_{\lambda(n)},k-1)+1-\Lambda(V_n,V_{\lambda (n)})
\end{equation}
\begin{equation}\label{baseequationofour}
   \chi (\mathbb{V}_n,k)=1
\end{equation}
 \begin{equation}\label{base01}
    \chi (\mathbb{V}_i,1)=1
\end{equation}
\begin{equation}\label{base02}
    \chi (\mathbb{V}_i,0)=0
\end{equation}

\end{document}